\documentclass[10pt,conference]{IEEEtran}\IEEEoverridecommandlockouts
\usepackage{epsfig,graphicx,subfigure,psfrag,amsmath,cases}
\usepackage{latexsym,amssymb,amsmath,epsfig,subfigure,algorithm,mathtools}
\usepackage{algorithmic}
\usepackage{color}
\usepackage{url}
\usepackage{scrtime}

\usepackage{subfigure}

\usepackage{multicol}

\author{Zhiqiang Wei,   Jiajia Guo, Derrick Wing Kwan Ng, and  Jinhong Yuan
\thanks{Zhiqiang Wei,  Jiajia Guo, Derrick Wing Kwan Ng, and Jinhong Yuan are with the School of Electrical
Engineering and Telecommunications, the University of New South Wales, Australia (email: zhiqiang.wei@student.unsw.edu.au; jiajia.guo@student.unsw.edu.au; w.k.ng@unsw.edu.au; j.yuan@unsw.edu.au). Derrick Wing Kwan Ng is
supported under Australian Research Council Discovery Early Career Researcher Award funding scheme (project number DE170100137). This work is partially supported by Australia Research Council (ARC) Discovery Project DP160104566.}
}

\title{Fairness Comparison of Uplink NOMA and OMA}

\newtheorem{Thm}{Theorem}

\newtheorem{T-Prob}{Transformed Problem}

\newtheorem{proof}{Proof}

\newtheorem{Remark}{Remark}

\newcommand{\abs}[1]{\lvert#1\rvert}

\begin{document}
\maketitle
\begin{abstract}
In this paper, we compare the resource allocation fairness of uplink communications between non-orthogonal multiple access (NOMA) schemes and orthogonal multiple access (OMA) schemes.
Through characterizing the contribution of the individual user  data rate to the system sum rate, we analyze the fundamental reasons that NOMA offers a more fair resource allocation than that of OMA in asymmetric channels.
Furthermore, a fairness indicator metric based on Jain's index is proposed to measure the asymmetry of multiuser channels.
More importantly, the proposed metric provides a selection criterion for choosing between NOMA and OMA for fair resource allocation.
Based on this discussion, we propose a hybrid NOMA-OMA scheme to further enhance the users fairness.
Simulation results confirm the accuracy of the proposed metric and demonstrate the fairness enhancement of the proposed hybrid NOMA-OMA scheme compared to the conventional OMA and NOMA schemes.
\end{abstract}
\renewcommand{\baselinestretch}{0.98}
\normalsize
\section{Introduction}
In the upcoming 5th generation (5G) wireless networks, non-orthogonal multiple access (NOMA) has been recognized as a promising consideration of multiple access scheme to accommodate more users and to improve the spectral efficiency\cite{Dai2015,Ding2015b,WeiSurvey2016,shin2016non,wong2017key}.
A preliminary version of NOMA, multiuser superposition transmission (MUST) scheme, has been proposed in the 3rd generation partnership project long-term evolution advanced (3GPP-LTE-A) networks\cite{Access2015}.
The principal idea of NOMA is to exploit the power domain for multiuser multiplexing and to utilize successive interference cancellation (SIC) to harness inter-user interference (IUI).
In contrast to conventional orthogonal multiple access (OMA) schemes \cite{Kwan_AF_2010,DerrickEEOFDMA}, NOMA enables simultaneous transmission of multiple users on the same degrees of freedom (DOF) via superposition coding with different power levels.
Meantime, by exploiting the received power disparity, advanced signal processing techniques, e.g., SIC, can be adopted to retrieve the desired signals at the receiver. It has been proved that NOMA can increase the system spectral efficiency substantially compared to the conventional OMA schemes\cite{Ding2014,Yang2016,Sun2016Fullduplex}. As a result, NOMA is able to support massive connections, to reduce communication latency, and to increase system spectral efficiency.

Most of existing works focused on  downlink NOMA systems \cite{Ding2014,Yang2016,Sun2016Fullduplex,Wei2016NOMA}.
However, NOMA inherently exists in uplink communications, where electromagnetic waves are naturally superimposed with different received power at a receiving base station (BS). Besides, SIC decoding is generally more affordable for BSs than mobile users.
The authors in \cite{Wang2006} compared NOMA and OMA in the uplink from the perspective of spectral-power efficiency.
Most recently, the authors in \cite{Al-Imari2014,Al-Imari2015} designed a resource allocation algorithm based on the maximum likelihood (ML) receiver at the BS.
On the other hand, another key feature of NOMA is to offer fairness provisioning in resource allocation. In contrast to OMA systems where users with poor channel conditions may temporarily suspended from service,   NOMA allows users with disparate channel conditions being served simultaneously. In \cite{ZhangUplinkNOMA}, a power allocation scheme was proposed to provide the max-min fairness to users in an uplink NOMA system. In \cite{TakedaUplinkFariness}, the authors studied a proportional fair based scheduling scheme for non-orthogonal multiplexed users.  In \cite{Timotheou2015,LiuFairnessNOMA}, power allocation with fairness consideration was investigated for single antenna and multiple antennas NOMA downlink systems, respectively. Despite some preliminary works \cite{ZhangUplinkNOMA,TakedaUplinkFariness,Timotheou2015,LiuFairnessNOMA,Diamantoulakis2016,yu2016antenna} have already considered fairness in resource allocation, it is still unclear why and when NOMA offers a more fair resource allocation than that of OMA.

In this paper, we aim to compare the fairness in resource allocation of uplink between NOMA and OMA.
To this end, a selection criterion is proposed for determining whether NOMA or OMA should be used given current channel state information.
Through characterizing the contribution of achievable data rate of individual users to the system sum rate, we explain the underlying reasons that NOMA is more fair in resource allocation than that of OMA in asymmetric channels.
Furthermore, for two-user NOMA systems\footnote{In a two-user NOMA system, there are at most two users multiplexing on the same DOF to reduce the computational complexity and delay incurred by SIC at receivers.}, we propose a closed-form fairness indicator metric to determine when NOMA is more fair than OMA.
In addition, a simple hybrid NOMA-OMA scheme which adaptively chooses NOMA and OMA according to the proposed metric is proposed to further enhance the users fairness.
Numerical results are shown to verify the accuracy of our proposed metric and to demonstrate the fairness enhancement of the proposed hybrid NOMA-OMA scheme.

The rest of the paper is organized as follows. In Section II, we present the uplink NOMA system model and discuss the capacity regions of NOMA and OMA.
In Section III, the reason of NOMA being more fair than OMA is analyzed. Besides, a closed-form fairness indicator metric and a hybrid NOMA-OMA scheme are proposed.
Simulation results are presented and analyzed in Section IV. Finally,
Section V concludes this paper.

Notations used in this paper are as follows. The circularly symmetric complex Gaussian distribution with mean $\mu$ and variance $\sigma^2$ is denoted by ${\cal CN}(\mu,\sigma^2)$;
$\sim$ stands for ``distributed as"; $\mathbb{C}$ denotes the set of all complex numbers; $\abs{\cdot}$ denotes the absolute value of a complex scalar; $\Pr \left\{  \cdot  \right\}$ denotes the probability of a random event.
\section{System Model}
In this section, we present an uplink NOMA system model and introduce the capacity regions of NOMA and OMA.
\subsection{System Model}
\begin{figure}[t]
\centering
\includegraphics[width=2in]{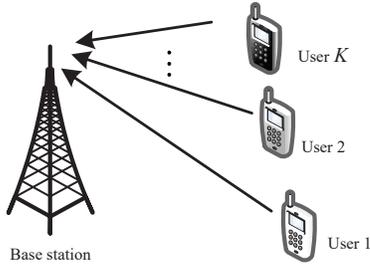}
\caption{The system model for uplink NOMA with one BS and $K$ users.}
\label{NOMA_Uplink_Model}
\end{figure}

We consider an uplink NOMA system with one single-antenna BS and $K$ single-antenna users, as shown in Figure \ref{NOMA_Uplink_Model}.
All the $K$ users are transmitting within a single subcarrier\footnote{To simplify the notations, we focus on the system with $K$ users multiplexing on a single subcarrier. This case will generalized to the case of multi-carrier systems in Section \ref{Hybrid} and Section \ref{Simulation}.} with the same maximum transmit power $P_0$.
For the NOMA scheme, $K$ users are multiplexed on the same subcarrier with different received power levels, while for the OMA scheme, $K$ users are utilizing the subcarrier via the time-sharing strategy \cite{Tse2005}.

For the NOMA scheme, the received signal at the BS is given by
\begin{align}\label{System Model}
y = \sum\limits_{k = 1}^K {\sqrt {p_k} {h_k}{s_k}}  + v.
\end{align}
where $h_k \in \mathbb{C}$ denotes the channel coefficient between the BS and user $k\in\{1,
\ldots,K\}$, $s_k$ denotes the modulated symbol for user $k$, $p_k$ denotes the transmit power of user $k$, and $v\sim{\cal CN}(0,\sigma^2)$ denotes the additive white Gaussian noise (AWGN) at the BS and $\sigma^2$ is the noise power. Without loss of generality, we assume that $\left|h_1\right|^2 \le \left|h_2\right|^2 \le \cdots \le \left|h_K\right|^2$.

\subsection{Capacity Region}\label{ResourceAllocation}
\begin{figure}[t]
\centering
\includegraphics[width=3.5in]{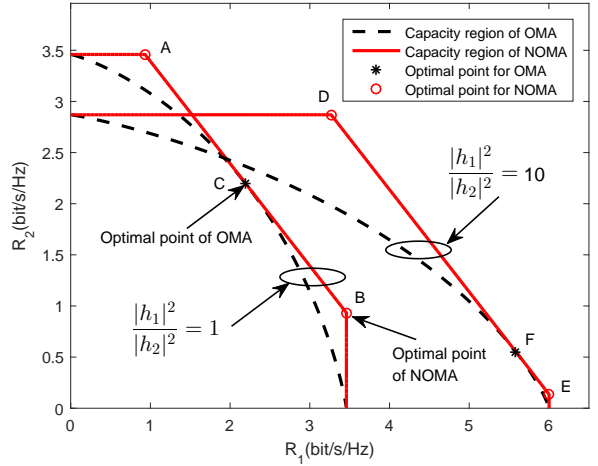}
\caption{The capacity region of NOMA and OMA with one BS and two users for a single channel realization.
The transmit power of both users is $P_0 = 20$ dBm. For the curve of $\frac{\left|h_1\right|^2}{\left|h_2\right|^2} = 1$, we have $\frac{\left|h_1\right|^2}{\sigma ^2} = \frac{\left|h_2\right|^2}{\sigma ^2} = 20$ dB. For the curve of $\frac{\left|h_1\right|^2}{\left|h_2\right|^2} = 10$, we have $\frac{\left|h_1\right|^2}{\sigma ^2} = 18$ dB and $\frac{\left|h_2\right|^2}{\sigma ^2} = 28$ dB.}
\label{CapacityRegion}
\end{figure}

It is well known that the OMA scheme with the optimal DOF allocation and the NOMA scheme with the optimal power allocation can achieve the same system sum rate in uplink transmission\cite{Tse2005,VaeziHK}, as shown in Figure \ref{CapacityRegion}.
Here, the optimal resource allocation for both NOMA and OMA schemes is in the sense of maximizing the system sum rate.
To facilitate the following presentation, we define $\alpha_k$ as a time-sharing factor for user $k$, where $\sum\limits_{k = 1}^K {\alpha_k} = 1$.
Particularly, the optimal DOF allocation of the OMA scheme, i.e., point C and point F in Figure \ref{CapacityRegion}, can be achieved by \cite{Tse2005}:
\begin{equation}
\alpha_k = \frac{\left|h_k\right|^2}{\sum\limits_{i = 1}^K {\left|h_i\right|^2}},\;\forall k.
\end{equation}
Note that $\alpha_k$ can also be interpreted as the normalized channel gain of user $k$.
In other words, the optimal DOF allocation for the OMA scheme is to share the subcarrier with the time duration proportional to their normalized channel gains, whereas it relies on adaptive time allocation according to the instantaneous channel realizations.
We note that the optimal DOF allocation is obtained with all the users transmitting with their maximum transmit power $P_0$ since there is no IUI in the OMA scheme.

On the other hand, power allocation of NOMA that achieves the corner points, i.e., point A, point B, point D, and point E in Figure \ref{CapacityRegion}, can be obtained by simply setting $p_k = P_0,\; \forall k$, and performing SIC at the BS\cite{Tse2005,Ali2016}.
Any rate pairs on the line segments between the corner points can be achieved via a time-sharing strategy.
It can be observed from Figure \ref{CapacityRegion} that NOMA with a time-sharing strategy always outperforms OMA, both in the sense of spectral efficiency and user fairness, since the capacity region of OMA is a subset of that of NOMA.
We note that NOMA without the time-sharing strategy can only achieve the corner points in the capacity region, which might be less fair than OMA in some cases.

In this paper, we study the users fairness of the NOMA scheme without time-sharing and the OMA scheme with an adaptive DOF allocation.
Both schemes achieve the same system sum rate but results in different users fairness.
Intuitively, in Figure \ref{CapacityRegion}, for symmetric channel with $\frac{\left|h_1\right|^2}{\left|h_2\right|^2} = 1$, OMA at point C is more fair than NOMA since both users have the same individual data rate. However, for an asymmetric channel with $\frac{\left|h_1\right|^2}{\left|h_2\right|^2} = 10$, it can be observed that NOMA at the optimal point D is more fair than OMA at the optimal point F.
Therefore, it is interesting to unveil the reasons for fairness enhancement of NOMA in asymmetric channels and to derive a quantitative fairness indicator metric for determining when NOMA is more fair than OMA.


\section{Fairness Comparison of NOMA and OMA}
In this section, we first present the adopted Jain's fairness index\cite{LanFairness} for quantifying the notion of resource allocation fairness.
Then, we characterize the contribution of individual user data rate to the system sum rate and investigate the underlying reasons of NOMA being more fair than OMA.
Subsequently, for a two-user NOMA system, a closed-form fairness indicator metric is derived from Jain's index \cite{LanFairness} to determine whether using NOMA or OMA for any pair of users on a single subcarrier.
Furthermore, a hybrid NOMA-OMA scheme is proposed which employs NOMA or OMA adaptively based on the proposed metric.

\subsection{Jain's Fairness Index}
In this paper, we adopt the Jain's index\cite{LanFairness} as the fairness measurement in the following
\begin{equation}\label{Jain}
J = \frac{\left( \sum\limits_{k = 1}^K R_k  \right)^2}{K\sum\limits_{k = 1}^K {\left( {R_k} \right)}^2 },
\end{equation}
where $R_k$ denotes the individual rate of user $k$.
Note that $\frac{1}{K}\le J\le1$. A scheme with a higher Jain's index is more fair and it achieves the maximum when all the users obtain the same individual data rate.

\subsection{Fairness Analysis}
\begin{figure}[t]
\centering
\includegraphics[width=3.5in]{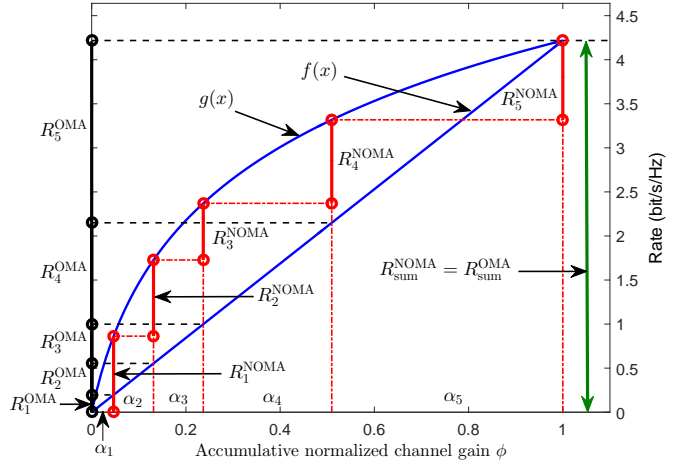}
\caption{An illustration of  system sum rate versus the accumulative normalized channel gains  for the NOMA and OMA with $K=5$ uplink users.  The sum rates of the NOMA scheme and the OMA scheme are denoted by the green double-side arrow.  The individual rates of the NOMA scheme and the OMA scheme are denoted by the red line segments and the black line segments, respectively.}
\label{LinearLog}
\end{figure}

For the optimal resource allocation of both NOMA and OMA schemes discussed in Section \ref{ResourceAllocation}, it is easily to obtain the sum rate and individual data rates for both schemes as follows:
\begin{align}
R_{\mathrm{sum}}^{\mathrm{NOMA}} &= R_{\mathrm{sum}}^{\mathrm{OMA}} = \sum\limits_{i = 1}^K R_{k}^{\mathrm{NOMA}} = \sum\limits_{i = 1}^K R_{k}^{\mathrm{OMA}} \notag\\
&= {\log _2}\left( 1 + \frac{P_0}{\sigma ^2} \sum\limits_{i = 1}^{K}{\left| {h_i} \right|}^2  \right),\label{SumRate}\\
R_{k}^{\mathrm{NOMA}} &= {\log _2}\left( 1 + \frac{P_0 \left| {h_k} \right|^2}{ P_0\sum\limits_{i = 1}^{k-1} {\left| {h_i} \right|}^2 + \sigma ^2}  \right),\; \text{and}\label{IndRate_NOMA}\\
R_{k}^{\mathrm{OMA}} &= {\alpha_k}R_{\mathrm{sum}}^{\mathrm{OMA}},\label{IndRate_OMA}
\end{align}
where $R_{\mathrm{sum}}^{\mathrm{NOMA}}$ and $R_{\mathrm{sum}}^{\mathrm{OMA}}$ denote the system sum rate for NOMA and OMA schemes with the optimal resource allocation, respectively, and $R_{k}^{\mathrm{NOMA}}$ and $R_{k}^{\mathrm{OMA}}$ denote the individual data rate for user $k$ in NOMA and OMA schemes, respectively.

For the NOMA scheme, we first define the accumulative normalized channel gain as $\phi_k = {\sum\limits_{i = 1}^{k} {\alpha_i}}$, $k = \left\{ {1, \cdots ,K} \right\}$, $\phi_0 = 0$, and then rewrite the achievable rate of user $k$ as
\begin{align}
R_{k}^{\mathrm{NOMA}}=& {\log _2}\left( 1 + \frac{{P_0}\phi_k}{\sigma ^2}\sum\limits_{i = 1}^K {\left| {h_i} \right|}^2  \right) \notag\\
& - {\log _2}\left( 1 + \frac{{P_0}\phi_{k-1}}{\sigma ^2}\sum\limits_{i = 1}^K {\left| {h_i} \right|}^2  \right).\label{IndRate_NOMA2}
\end{align}
The first term in \eqref{IndRate_NOMA2} denotes the sum rate of a system with $k$ users and the second term denotes the counterpart of a system with $k-1$ users.
In other words, the contribution of user $k$ to the system sum rate depends on the difference of a logarithm function with respect to (w.r.t.) $\phi_k$ and $\phi_{k-1}$. For notational simplicity and without loss of generality, we define the logarithm function as
\begin{equation}
g\left(x\right) = {\log _2}\left( 1 + \Gamma x  \right),\; 0\le x \le 1,
\end{equation}
with
\begin{align}\label{Gamma}
\Gamma = \frac{{P_0}}{\sigma ^2}\sum\limits_{i = 1}^K {\left| {h_i} \right|}^2 \;\; \text{and} \;\;
R_{k}^{\mathrm{NOMA}} = g\left(\phi_k\right) - g\left(\phi_{k-1}\right).
\end{align}

On the other hand, for the OMA scheme, it can be observed from \eqref{IndRate_OMA} that $R_{k}^{\mathrm{OMA}}$ has a linear relationship with $R_{\mathrm{sum}}^{\mathrm{OMA}}$ and the slope w.r.t. the system sum rate is determined by the normalized channel gain ${\alpha_k} = \phi_k - \phi_{k-1}$. Similarly, the contribution of user $k$ to the system sum rate depends on the difference of a linear function of $\phi_k$ and $\phi_{k-1}$, where the linear function is given by
\begin{align}
f\left(x\right) &= {\log _2}\left( 1 + \Gamma \right)x,\; 0\le x \le 1, \;\; \text{and} \notag\\
R_{k}^{\mathrm{OMA}} &= f\left(\phi_k\right) - f\left(\phi_{k-1}\right).
\end{align}

Figure \ref{LinearLog} illustrates the linear and logarithmic increments of the system data rate w.r.t. the accumulative channel gain for OMA and NOMA, respectively, with $K = 5$ uplink users. It can be observed that the NOMA and OMA schemes have the same system sum rate but contributed by different date rates of individual users. In particular, the NOMA scheme achieves a more fair resource allocation than that of the OMA scheme since all the users are allocated with similar individual rates. In fact, the fairness of resource allocation in NOMA inherits from the logarithmic mapping of $g\left(\phi_k\right)$ w.r.t. the accumulative channel gain $\phi_k$. The first and second derivatives of $g(\phi_k)$ are increasing and decreasing w.r.t. $\phi_k$, respectively.  The larger normalized channel gain $\alpha_k$, the slower $g\left(\phi_k\right)$ increasing with $\phi_k$, which results in a smaller individual rate compared to that of the OMA scheme.
On the other hand, a smaller normalized channel gain $\alpha_k$ would result in a higher increasing rate of $g\left(\phi_k\right)$ with $\phi_k$, when a higher individual rate is obtained compared to that of the OMA scheme. For instance, considering the weakest user and the strongest user with their normalized channel gain $\alpha_1$ and $\alpha_K$, respectively,
$R_{1}^{\mathrm{NOMA}}$ is raised up by the logarithm function $g\left(x\right)$ compared to $R_{1}^{\mathrm{OMA}}$, while $R_{K}^{\mathrm{NOMA}}$ is reduced compared to $R_{K}^{\mathrm{OMA}}$.

\begin{Remark}

Note that for  symmetric channels, linear mapping of the OMA scheme is more fair than the NOMA scheme. However, the probability that all the users have the same channel gains is quite small, especially for a system with a large number of users.
\end{Remark}
\subsection{Fairness Indicator Metric}
In practice, most of NOMA schemes assume that there are at most two users multiplexing via the same DOF\cite{Dingtobepublished,Sun2016Fullduplex,Wei2016NOMA}, which can reduce both the computational complexity and decoding delay at the receiver.
Therefore, we focus on the fairness comparison of NOMA and OMA with $K=2$ in this section.
We aim to find a simple metric to determine when NOMA is more fair than OMA for any pair of users, which is fundamentally important for user scheduling design in the system with multiple DOF and multiple users. The fairness indicator metric is proposed in the following theorem.

%

\begin{Thm}\label{Theorem1}
Given a pair of users with their channel realizations $\left|h_1\right|^2 \le \left|h_2\right|^2$, the NOMA scheme is more fair in the sense of Jain's fairness index if and only if
\begin{equation}
\frac{\left|h_1\right|^2}{\left|h_2\right|^2} \le \frac{\beta}{1-\beta},
\end{equation}
where $\beta = \frac{W\left( \frac{{\left( {1 + \Gamma } \right)}^{1 + \frac{1}{\Gamma }}\log (1 + \Gamma )}{\Gamma} \right)}{\log (1 + \Gamma )} - \frac{1}{\Gamma }$ and $W(x)$ is the Lambert W function.
In the high SNR regime, i.e., $\Gamma \to \infty $, we have the high SNR approximation of $\beta$ as
\begin{equation}\label{HSNRA}
\widetilde{\beta} \approx \frac{W\left( \log (1 + \Gamma ) \right)}{\log (1 + \Gamma )}.
\end{equation}
\end{Thm}
\begin{proof}
Since both the NOMA and OMA schemes have the same sum rate, we need to compare the sum of square of individual rates ($\mathrm{SSR}$), i.e., $\mathrm{SSR} = \sum\limits_{k = 1}^2 {\left( {R_k} \right)}^2$, in the denominator of \eqref{Jain}. The scheme with a smaller $\mathrm{SSR}$ would be more fair in terms of Jain's index. For the OMA scheme, we have
\begin{align}
\mathrm{SSR}_{\mathrm{OMA}} &= \left( \log_2 (1 + \Gamma )\right)^2 \left( {\alpha _1^2}+{\alpha _2^2} \right)\notag\\
& = \left( \log_2 (1 + \Gamma )\right)^2 \left( 1 + 2{\alpha _1^2} -2 \alpha _1 \right),
\end{align}
where $0 \le \alpha _1 \le 0.5$ since we assume $\left|h_1\right|^2 \le \left|h_2\right|^2$.
For the NOMA scheme, the $\mathrm{SSR}_{\mathrm{NOMA}}$ can be given by
\begin{align}
\mathrm{SSR}_{\mathrm{NOMA}}=& \left( \log_2 (1 \hspace*{-1mm}+\hspace*{-1mm} \Gamma \alpha _1)\right)^2 \hspace*{-0.2mm}+\hspace*{-0.2mm} \left( \log_2 (1 \hspace*{-1mm}+\hspace*{-1mm} \Gamma) \hspace*{-0.2mm}-\hspace*{-0.2mm} \log_2 (1 \hspace*{-1mm}+\hspace*{-1mm} \Gamma \alpha _1) \right)^2 \notag\\
=& \left( \log_2 (1 + \Gamma )\right)^2 + 2\left( \log_2 (1 + \Gamma \alpha _1)\right)^2 \notag\\
& - 2 \log_2 (1 + \Gamma ) \log_2 (1 + \Gamma \alpha _1).
\end{align}

Note that a trivial solution for $\mathrm{SSR}_{\mathrm{OMA}}=\mathrm{SSR}_{\mathrm{NOMA}}$ is given with $\alpha _1 = 0$, which corresponds to a single user scenario.
In addition, at $\alpha _1 = 0.5$, i.e., $\left|h_1\right|^2 = \left|h_2\right|^2$, we have $\mathrm{SSR}_{\mathrm{OMA}} < \mathrm{SSR}_{\mathrm{NOMA}}$ as observed from the capacity region in Figure \ref{CapacityRegion}.
Further, $\mathrm{SSR}_{\mathrm{OMA}}$ is a monotonic decreasing function of $\alpha _1$ within $0 \le \alpha _1 \le 0.5$, while $\mathrm{SSR}_{\mathrm{NOMA}}$ is a monotonic decreasing function of $\alpha _1$ within $0 \le \alpha _1 \le \frac{\sqrt {1 + \Gamma }  - 1}{\Gamma }$ and it is increasing with $\alpha _1$ within $\frac{\sqrt {1 + \Gamma }  - 1}{\Gamma } \le \alpha _1 \le 0.5$. Also, from Figure \ref{CapacityRegion}, we can observe that $\mathrm{SSR}_{\mathrm{OMA}} > \mathrm{SSR}_{\mathrm{NOMA}}$ for an arbitrary small positive $\alpha _1$.
Therefore, there is a unique intersection of $\mathrm{SSR}_{\mathrm{OMA}}$ and $\mathrm{SSR}_{\mathrm{NOMA}}$ at $\alpha _1 = \beta$ in the range of $0 < \alpha _1 < 0.5$.
Before the intersection, i.e., $0 < \alpha _1 < \beta$, NOMA is more fair, while after the intersection, i.e., $\beta < \alpha _1 < 0.5$, OMA is more fair.
Solving the equation of $\mathrm{SSR}_{\mathrm{OMA}} = \mathrm{SSR}_{\mathrm{NOMA}}$ within $0 < \alpha _1 < 0.5$, we obtain
\begin{equation}
\beta = \frac{W\left( \frac{{\left( {1 + \Gamma } \right)}^{1 + \frac{1}{\Gamma }}\log (1 + \Gamma )}{\Gamma} \right)}{\log (1 + \Gamma )} - \frac{1}{\Gamma }.
\end{equation}
Furthermore, with $\alpha _1 \le \beta$, we have $\frac{\left|h_1\right|^2}{\left|h_2\right|^2} \le \frac{\beta}{1-\beta}$, which completes the proof for the sufficiency of the proposed fairness indicator metric.

For the necessity, since the intersection of $\mathrm{SSR}_{\mathrm{OMA}}$ and $\mathrm{SSR}_{\mathrm{NOMA}}$ within $0 < \alpha _1 < 0.5$ is unique, the only region within $0 < \alpha _1 < 0.5$ where $\mathrm{SSR}_{\mathrm{OMA}} > \mathrm{SSR}_{\mathrm{NOMA}}$ is $0 < \alpha _1 < \beta$. In other words, NOMA is more fair only if $0 < \alpha _1 < \beta$, which completes the proof for the necessity of the proposed metric.
\end{proof}

\begin{Remark}\label{Remark1}
Note that the proposed fairness indicator metric only depends on the parameter $\Gamma$ defined in \eqref{Gamma}.
As a result, the metric depends on the instantaneous channel gains.
Compared to the Jain's index, our proposed metric is more insightful which connects OMA and NOMA.
Particularly, for the high SNR approximation \eqref{HSNRA}, we can observe that $\widetilde{\beta}$ decreases with the increasing maximum transmit power since the Lambert W function in the numerator increases slower than that of the denominator.
Therefore, the probability of NOMA being more fair will decrease when increasing the maximum transmit power, which will be verified in the simulations.
\end{Remark}
\subsection{A Hybrid NOMA-OMA Scheme}\label{Hybrid}
The proposed fairness indicator metric in Theorem \ref{Theorem1} provides a simple way to determine if NOMA is more fair than OMA, and would serve as a criterion for user scheduling design for systems with multi-carrier serving multiple users. In particular, for an arbitrarily user scheduling strategy,  we propose an adaptive hybrid scheme which decides each pair of users on each subcarrier in choosing either the OMA scheme or the NOMA scheme  to enhance users fairness.
Instead of using the NOMA scheme or the OMA scheme across all the subcarriers, this hybrid NOMA-OMA scheme can enhance the user fairness substantially.
Note that the fairness performance can be further improved if it is jointly designed with the user scheduling. It will be considered in the future work.

\section{Simulation Results}\label{Simulation}
In this section, we adopt simulations to verify the effectiveness of the proposed metric and to evaluate the proposed hybrid NOMA-OMA scheme.
A single cell with a BS located at the center with a cell radius of $400$ m is considered.
There are $N_\mathrm{F} = 128$ subcarriers in the system and $2N_\mathrm{F}$ numbers of users are randomly paired on all the subcarriers.
All the $2N_\mathrm{F}$ users are randomly and uniformly distributed in the cell. We set the noise power in each subcarrier at the BS as $\sigma ^2 = -90$ dBm. The 3GPP path loss model in urban macro cell scenario is adopted in our simulations\cite{Access2010}.

\begin{figure}[t]
\centering
\includegraphics[width=3.5in]{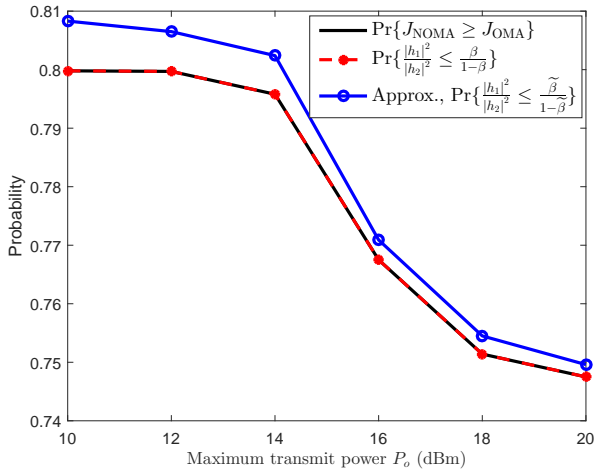}
\caption{The probability of NOMA being more fair than OMA versus the maximum transmit power, $P_0$.}
\label{PredictedActual}
\end{figure}

\begin{figure}[t]
\centering
\includegraphics[width=3.5in]{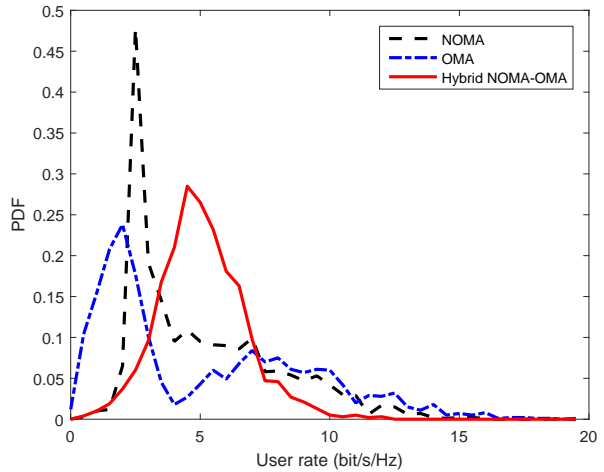}
\caption{The PDF of user rate for the NOMA scheme, the OMA scheme, and the hybrid NOMA-OMA scheme.}
\label{PDF_NOMA_OMA}
\end{figure}

\begin{figure}[t]
\centering
\includegraphics[width=3.5in]{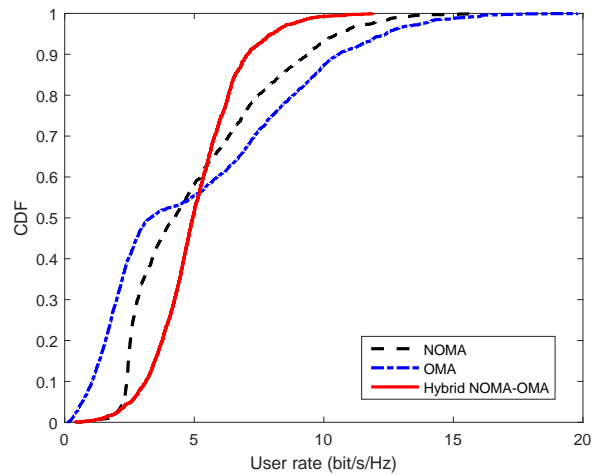}
\caption{The CDF of the NOMA scheme, the OMA scheme, and the hybrid NOMA-OMA scheme.}
\label{CDF_NOMA_OMA}
\end{figure}

Figure \ref{PredictedActual} depicts the probability of NOMA being more fair than OMA versus the maximum transmit power, $P_0$.
It can be observed that $\mathrm{Pr}\left\{ \frac{\left|h_1\right|^2}{\left|h_2\right|^2} \le \frac{\beta}{1-\beta} \right\}$ matches well with $\mathrm{Pr}\left\{ J_{\mathrm{NOMA}} \ge J_{\mathrm{OMA}} \right\}$. In other words, our proposed fairness indicator metric can accurately predict if NOMA is more fair than OMA.
Also, for the high SNR approximation $\widetilde{\beta}$ in \eqref{HSNRA}, $\mathrm{Pr}\left\{ \frac{\left|h_1\right|^2}{\left|h_2\right|^2} \le \frac{\widetilde{\beta}}{1-\widetilde{\beta}}\right\}$ closely matches with the simulation results.
In addition, we can observe that the NOMA scheme has a high probability ($0.75\sim0.8$) of being more fair than that of the OMA scheme in terms of Jain's index.
This is due to the fact that the probability of asymmetric channels is much larger than that of symmetric channels.
On the other hand, the probability of NOMA being more fair is decreasing with the maximum transmit power as discussed in Remark 2.
This is because the NOMA scheme is interference-limited in the high transmit power regime.
Specifically, the strong user (with higher received power) will face a large amount of interference, while the weak user (with lower received power) is interference-free owing to the SIC decoding.
As a result, in the high transmit power regime, the weak user can achieve a much higher data rate than that of the strong user, which may result in a less fair resource allocation than that of OMA.
Even though, NOMA is still more fair than OMA with a probability of about $0.75$ in the high transmit power regime.

Figure \ref{PDF_NOMA_OMA} shows the probability density function (PDF) of user rate for a multi-carrier system with a random pairing strategy.
Three multiple access schemes are compared, including the NOMA scheme, the OMA scheme, and the proposed hybrid NOMA-OMA scheme.
It can be observed that the individual data rate distribution of the NOMA scheme is more concentrated than that of the OMA scheme, which means that the NOMA scheme offers a more fair resource allocation than the OMA scheme.
Further, the individual rate distribution of the hybrid NOMA-OMA scheme is more concentrated than that of the NOMA scheme.
In fact, our proposed hybrid NOMA-OMA scheme can better exploit the channel gains' relationship via the adaptive selection between NOMA and OMA according to the fairness indicator metric.
Actually, for the three multiple access schemes, we have $J_{\mathrm{NOMA}} = 0.76$ $J_{\mathrm{OMA}} = 0.62$, and $J_{\mathrm{Hybrid}} = 0.91$, where $J_{\mathrm{Hybrid}}$ denotes the Jain's index for the hybrid NOMA-OMA scheme.
In addition, the cumulative distribution function (CDF) of user rate is more of interest in practice, which is illustrated in Figure \ref{CDF_NOMA_OMA}.
We can observe that the 10th-percentile the user rate, which is closely related to fairness and user experience, increased about $1$ bit/s/Hz compared to that of the NOMA scheme.
This shows that our proposed hybrid NOMA-OMA scheme can significantly improve the performance of low-rate users and therefore elevate the quality of user experience.

\section{Conclusion}
In this paper, we investigated the resource allocation fairness of the NOMA and OMA schemes in uplink.
The fundamental reason of NOMA being more fair than OMA in asymmetric multiuser channels was analyzed through characterizing the contribution of data rate of each user to the system sum rate.
It is the logarithmic mapping between the normalized channel gains and the individual data rates that exploits the channel gains asymmetry to enhance the users fairness in the NOMA scheme.
Based on this observation, we proposed a quantitative fairness indicator metric for two-user NOMA systems which determines if NOMA offers a more fair resource allocation than OMA.
In addition, we proposed a hybrid NOMA-OMA scheme that adaptively choosing between NOMA and OMA based on the proposed metric to further improve the users fairness.
Numerical results demonstrated that our proposed metric can accurately predict when NOMA is more fair than OMA.
Besides, compared to the conventional NOMA and OMA schemes, the proposed hybrid NOMA-OMA scheme can substantially enhance the users fairness.

\end{document}